\documentclass{article}
\usepackage{amsmath,amsthm,amssymb}
\usepackage{array, longtable}
\newtheorem{Theorem}{Theorem}[section]
\newtheorem{lem}[Theorem]{Lemma}

\newtheorem{Definition}[Theorem]{Definition}

\newtheorem{Example}[Theorem]{Example}

\numberwithin{equation}{section}

\begin{document}

\title{A Construction of MDS Quantum Convolutional Codes\footnote{
 {\small Email addresses: zghui2012@126.com (G. Zhang),  bocong\_chen@yahoo.com (B. Chen), lcli@yahoo.com (L. Li).}}}

\author{Guanghui Zhang$^1$, Bocong Chen$^2$, Liangchen Li$^1$}

\date{\small
${}^1$School of Mathematical Sciences,
Luoyang Normal University,
Luoyang, Henan, 471022, China\\
${}^2$School of Physical \& Mathematical Sciences,
         Nanyang Technological University, Singapore 637616, Singapore}

\maketitle
\begin{abstract}
In this paper, two new families of MDS quantum convolutional codes are constructed.
The first one can be regarded as a generalization of \cite[Theorem 6.5]{GGGlinear}, in the sense that
we do not assume that $q\equiv1\pmod{4}$.
More specifically, we obtain two classes of MDS quantum convolutional codes with parameters:
{\rm (i)}~ $[(q^2+1, q^2-4i+3,1;2,2i+2)]_q$, where $q\geq5$ is an odd prime power and $2\leq i\leq(q-1)/2$;
{\rm (ii)}~ $[(\frac{q^2+1}{10},\frac{q^2+1}{10}-4i,1;2,2i+3)]_q$, where $q$ is an odd prime power with the form
 $q=10m+3$ or $10m+7$ ($m\geq2$),  and $2\leq i\leq2m-1$.


\medskip
\textbf{PACS number(s): 03.67.Pp, 89.70.-a}
\end{abstract}

\section{Introduction}
Quantum block codes  are used to protect   quantum information over noisy quantum channels.
Many works have been done for the constructions of good quantum error-correcting codes (e.g. see \cite{Steane}-\cite{Chen14}).
Quantum convolutional coding theory provides  a different paradigm for coding quantum information
and has numerous benefits for quantum communication  (\cite{Ollivier1}-\cite{FGG07}). For example, the convolutional structure is useful for a quantum communication scenario where a sender possesses a stream of qubits to send to a receiver.

The first important quantum block code construction is that of \cite{Steane}-\cite{Steane2},
which yields  the  commonly called  Calderbank Shor Steane (CSS) construction.
In contrast to quantum block codes, the construction for a CSS quantum convolutional code is similar to that for the block
case, except that we import classical convolutional codes rather than classical block codes \cite[Chap 9]{LB}.
Forney {\it et al.} \cite{FGG07} provided  many constructions of CSS quantum convolutional
codes from classical binary convolutional codes.
A Calderbank-Rains-Shor-Sloane (CRSS) quantum convolutional code was obtained
from a classical convolutional code over $\mathbb{F}_4$ \cite[Chap 9]{LB}.
Many classes of quantum convolutional codes
have been constructed (e.g. see \cite{GR06}-\cite{GGGarxiv}).

Almeida
and Palazzo Jr. in \cite{DP} obtained a
quantum convolutional code with parameters  $[(4, 1, 3)]$  (memory $m=3$).
Tan and Li in \cite{TL} constructed quantum convolutional codes
through  LDPC codes.
Very recently,  La Guardia \cite{GGG14}-\cite{GGGarxiv}  applied  the  methods   presented
by Piret in \cite{Piret} and then generalized  by Aly {\it et al.} in \cite{AGKM}, to
construct classical and MDS quantum convolutional codes.

Motivated by \cite{GGGlinear}, two new families of MDS quantum convolutional codes are constructed in
this paper. The first one can be regarded as a generalization of \cite[Theorem 6.5]{GGGlinear}, in the sense that
we do not assume that $q\equiv1\pmod{4}$.
More specifically, we obtain two classes of MDS quantum convolutional codes with parameters:
{\rm (i)}~$[(q^2+1, q^2-4i+3,1;2,2i+2)]_q$, where $q\geq5$ is an odd prime power and $2\leq i\leq(q-1)/2$;
{\rm (ii)}~$[(\frac{q^2+1}{10},\frac{q^2+1}{10}-4i,1;2,2i+3)]_q$, where $q$ is an odd prime power with the form
 $q=10m+3$ or $10m+7$ ($m\geq2$),  and $2\leq i\leq2m-1$.

The paper is organized as follows. In Section 2, we  recall basic notation and necessary  facts about constacyclic codes,
classical convolutional codes and  MDS quantum convolutional codes.
In Section 3, we propose constructions of new families of  MDS quantum convolutional codes derived from constacyclic  codes.

\section{Background}

In this section, we  recall basic notation and necessary  facts which
are important to the constructions of quantum convolutional codes.
We adopt the notation in \cite{GGGlinear}.

\subsection{Classical convolutional codes}
As mentioned in Section 1, quantum convolutional codes can be constructed from classical convolutional codes.
In this subsection we present a brief review of classical convolutional codes.
Let  $G(D)=(g_{ij})\in\mathbb{ F}_{q^2}[D]^{k\times n}$, where $\mathbb{ F}_{q^2}[D]^{k\times n}$ denotes
the set of all $k\times n$ matrices with entries in $\mathbb{F}_{q^2}[D]$;
$G(D)$ is called {\it basic} if it has a polynomial right inverse. A basic
generator matrix is called {\it reduced} if the overall constraint length
$\gamma=\sum_{i=1}^k\gamma_i$
has the smallest value among all basic generator matrices,  where $\gamma_i=\max_{1\leq j\leq n}{\{\deg g_{ij}\}}$.
In this case the overall constraint length
$\gamma$ will be called the {\it degree} of the resulting code.

\begin{Definition}(See \cite{AKS})
A  convolutional code $V$ with parameters $(n,k,\gamma; \mu, d_f)_{q^2}$ is a submodule of $\mathbb{F}_{q^2}[D]^n$
generated by a reduced basic matrix $G(D)=(g_{ij})\in\mathbb{ F}_{q^2}[D]^{k\times n}$,
$V=\{\mathbf{u}(D) G(D)\,|\,\mathbf{ u}(D)\in\mathbb{ F}_{q^2}[D]^k\}$,
where $n$ is the length, $k$ is the dimension, $\gamma_i=\max_{1\leq j\leq n}{\{\deg g_{ij}\}}$ is the degree,
$\mu=\max_{1\leq i\leq k}{\{\gamma_i\}}$ is the memory and
$d_f={\rm wt}(V)=\mid{\{{\rm wt}(\mathbf{v}(D))\,|\,\mathbf{v}(D)\in V, \mathbf{v}(D)\neq0\}}$
is the free distance of the code.
Here,  ${\rm wt}(\mathbf{v}(D))=\sum_{i=1}^n{\rm wt}(v_i(D))$, where ${\rm wt}(v_i(D))$
is the number of nonzero coefficients of $v_i(D)$.
\end{Definition}

The Hermitian inner product on $\mathbb{F}_{q^2}[D]^n$ is defined as
$\langle \mathbf{u}(D)\,|\,\mathbf{v}(D)\rangle_h=\sum_i\mathbf{u}_i\cdot\mathbf{v}_i^q$,
where $\mathbf{u}_i, \mathbf{v}_i\in\mathbb{ F}_{q^2}^n$ and $\mathbf{v}_i^q=(v_{1i}^q, v_{2i}^q, \cdots, v_{ni}^q)$.
The Hermitian dual of the code $V$ is defined by
$$
V^{\perp_h}=\Big\{\mathbf{ u}(D)\in\mathbb{ F}_{q^2}[D]^n\,\Big{|}\,\langle \mathbf{u}(D)\,|\,\mathbf{v}(D)\rangle_h=0 ~~\hbox{for~ all~ $\mathbf{v}(D)\in V$}\Big\}.
$$

We can construct  convolutional codes from block codes. Let $\mathcal{C}$ be an $[n,k,d]_{q^2}$
linear code with parity check matrix $H$. Split $H$ into $\mu+1$ disjoint submatrices $H_i$ such that
\begin{equation}\label{partition}
H=\left(
    \begin{array}{c}
      H_0 \\
      H_1\\
      \vdots \\
      H_\mu \\
    \end{array}
  \right)
\end{equation}
where each $H_i$ has $n$ columns. We then have the polynomial matrix
\begin{equation}\label{G}
G(D)=\tilde{H}_0+\tilde{H}_1D+\cdots+\tilde{H}_\mu D^\mu
\end{equation}
where the matrices $\tilde{H}_i$ for all $1\leq i\leq\mu$, are derived from the respective matrices
$H_i$ by adding zero-rows at the bottom in such a way that the matrix $\tilde{H}_i$
has $\kappa$ rows in total. Here $\kappa$  is the maximal number of rows among the matrices $H_i$,
$1\leq i\leq\mu$.
It is well known that $G(D)$ generates a convolutional
code with $\kappa$ rows, and that $\mu$ is the memory of the resulting
convolutional code.

\begin{Theorem}(See \cite[Theorem 3]{AGKM})\label{juanji}
Suppose that $\mathcal{C}$
is a
linear code over $\mathbb{F}_{q^2}$ with parameters $[n, k, d]_{q^2}$ and assume also that
$H\in\mathbb{ F}_{q^2}^{(n-k)\times n}$ is a parity check matrix for $\mathcal{C}$
partitioned into
submatrices $H_0, H_1, \cdots, H_\mu$ as in (\ref{partition})
such that $\kappa={\rm rk}H_0$ and ${\rm rk}H_i\leq\kappa$ for $1\leq i\leq\mu$
and consider the polynomial
matrix $G(D)$ as in (\ref{G}).
Then we have:

\item[{(1)}]The matrix $G(D)$ is a reduced basic generator matrix.

\item[{(2)}]If $\mathcal{C}^{\perp_h}\subseteq\mathcal{C}$, then the convolutional code
$V=\{\mathbf{u}(D) G(D)\,|\,\mathbf{ u}(D)\in\mathbb{ F}_{q^2}[D]^{n-k}\}$ satisfies $V\subseteq V^{\perp_h}$.

\item[{(3)}]If $d_f$ and $d_f^{\perp_h}$ denote the free distances of $V$ and $V^{\perp_h}$ respectively,
$d_i$ denotes the minimum distance  of the code $\mathcal{C}_i=\{\mathbf{v}\in \mathbb{F}_{q^2}^n\,|\, \mathbf{v}\tilde{H}_i^t=0\}$
and $d^{\perp_h}$ is the minimum distance of $\mathcal{C}^{\perp_h}$, then one has $\min{\{d_0+d_\mu, d\}}\leq d_f^{\perp_h}\leq d$
and $d_f\geq d^{\perp_h}$.
\end{Theorem}
Theorem \ref{juanji} suggests that one can obtain classical convolutional codes through linear codes over $\mathbb{F}_{q^2}$.
Constacyclic codes constitute a remarkable generalization of cyclic codes, hence form an important
class of linear codes in the coding theory. In this paper, we  apply Theorem \ref{juanji} to constacyclic codes.
The necessary notations and results about constacyclic codes are reviewed in the next subsection.

\subsection{Constacyclic codes}
Since we will work with codes endowed with the Hermitian inner product, we need to consider codes over $\mathbb{F}_{q^2}$,
where   $\mathbb{F}_{q^2}$ denotes
the finite field with $q^2$ elements. Let $\mathbb{F}_{q^2}^*=\mathbb{F}_{q^2}\setminus\{0\}$.
For $\lambda\in \mathbb{F}_{q^2}^*$,  we denote by $r=\rm{ord}(\lambda)$
the order of $\lambda$ in the cyclic group $\mathbb{F}_{q^2}^*$, i.e., $r$ is the smallest positive  integer  such that $\lambda^r=1$.
Then $r$ is a divisor of $q^2-1$, and $\lambda$
is called a {\em primitive $r$th root of unity}.

Starting from this section till the end of this paper, we assume that $n$ is a positive integer relatively prime to $q$.
A {\it $\lambda$-constacyclic code} $\mathcal{C}$ of length $n$ over $\mathbb{F}_{q^2}$ is an ideal of
the quotient ring $\mathbb{F}_{q^2}[X]/\langle X^n-\lambda\rangle$, where $\lambda\in \mathbb{F}_{q^2}^*$ (e.g.,  see \cite{Dinh12} or \cite{Chen1}).
It is well known that a unique monic polynomial $g(X)\in \mathbb{F}_{q^2}[X]$ can be found such that $g(X)\mid(X^n-1)$
and $\mathcal{C}=\langle g(X)\rangle=\{f(X)g(X)\,|\,f(X)\in \mathbb{F}_{q^2}[X]\}$.
In this case, $g(X)$ is called the {\it generator polynomial} of $\mathcal{C}$.

Assume  that   $\lambda\in \mathbb{F}_{q^2}^*$ is a primitive $r$th root of unity.
As mentioned before,  $r$ is a divisor of $q^2-1$. In particular, $\gcd(r,q)=1$, so $\gcd(rn,q)=1$.
We denote by $\ell={\rm ord}_{rn}(q^2)$, i.e.,  $\ell$ is the smallest positive integer such that $rn\mid(q^{2\ell}-1)$.
Then there exists a primitive $rn$th
root of unity $\beta\in \mathbb{F}_{q^{2\ell}}$ such that $\beta^n=\lambda$.
The roots of $X^n-\lambda$ are precisely the elements  $\beta^{1+ri}$ for $0\leq i\leq n-1$.
Set $\theta_{r,n}=\{1+ri\,|\,0\leq i\leq n-1\}$.
The {\it defining set} of a constacyclic code $\mathcal{C}=\langle g(X)\rangle$ of length $n$ is the set
$Z=\{j\in \theta_{r,n}\,|\,\beta^j ~\hbox{is a root of} ~g(X)\}$.
It is easy to see that the defining set $Z$ is  a union of some $q^2$-cyclotomic cosets modulo $rn$ and
$\dim_{\mathbb{F}_{q^2}}(\mathcal{C})=n-|Z|$ (see \cite{Yang} or \cite{Kai14}).
Since $\ell={\rm ord}_{rn}(q^2)$, it follows that the size of each $q^2$-cyclotomic cosets modulo $rn$
is a divisor of $\ell$ (e.g. see \cite[Theorem 4.1.4]{huff}).

The following
theorem gives the BCH bound for constacyclic codes (see  \cite[Theorem 4.1]{Yang}).
\begin{Theorem}({\bf The BCH bound for constacyclic codes})\label{BCH}
Let $\mathcal{C}$ be a $\lambda$-constacyclic code of length $n$ over $\mathbb{F}_{q^2}$, where $\lambda\in \mathbb{F}_{q^2}$ is a primitive $r$th root of unity.
Suppose $\ell={\rm ord}_{rn}(q^2)$.
Let $\beta\in \mathbb{F}_{q^{2\ell}}$ be  a primitive $rn$th root of unity  such that $\beta^n=\lambda$.
Assume that the generator polynomial of $\mathcal{C}$ has roots that include the set $\{\beta\zeta^i\,|\, i_1\leq i\leq i_1+d-2\}$, where
$\zeta=\beta^r$. Then the minimum distance of $\mathcal{C}$ is at least $d$.
\end{Theorem}

The {\it Hermitian inner product} on $\mathbb{F}_{q^2}^n$
is defined as
$$
(\mathbf{x},\mathbf{y})_h=x_0y_0^q+x_1y_1^q+\cdots+x_{n-1}y_{n-1}^q,
$$
where $\mathbf{x}=(x_0,x_1,\cdots,x_{n-1})\in \mathbb{F}_{q^2}^n$ and  $\mathbf{y}=(y_0,y_1,\cdots,y_{n-1})\in \mathbb{F}_{q^2}^n$.
For a linear code $\mathcal{C}$ of length $n$ over $\mathbb{F}_{q^2}$, the {\it Hermitian dual code} of $\mathcal{C}$ is defined as
$$
\mathcal{C}^{\perp_{h}}=\Big\{\mathbf{x}\in \mathbb{F}_{q^2}^n\,\Big{|}\,\sum\limits_{i=0}^{n-1}x_iy_i^q=0,~~~\hbox{for ~any~$\mathbf{y}\in \mathcal{C}$}\Big\}.
$$
If $\mathcal{C}\subseteq \mathcal{C}^{\perp_{h}}$, then $\mathcal{C}$ is called a
(Hermitian)   self-orthogonal code.
Conversely, if $\mathcal{C}^{\perp_{h}}\subseteq \mathcal{C}$, we say that $\mathcal{C}$ is a (Hermitian) dual-containing code.
For a $\lambda$-constacyclic code $\mathcal{C}$ of length $n$ over $\mathbb{F}_{q^2}$, it is shown that
$\mathcal{C}^{\perp_h}$ is a $\lambda^{-q}$-constacyclic
code; further, $\lambda=\lambda^{-q}$ precisely when $r\mid (q+1)$ (\cite[Lemma 2.1(ii)]{Yang}).

The following results are   useful.

\begin{lem}(See \cite[Lemma 2.2]{Kai14})\label{dual-containong}
Let $\lambda\in \mathbb{F}_{q^2}^*$ be a primitive $r$th root of unity. Assume that $\mathcal{C}$ is a $\lambda$-constacyclic
code of length $n$ over $\mathbb{F}_{q^2}$ with defining set $Z$. Then $\mathcal{C}$ is a dual-containing code if and only if $Z\bigcap Z^{-q}=\emptyset$,
where $Z^{-q}=\{-qz ~(\bmod ~rn)\,|\,z\in Z\}$.
\end{lem}

\begin{lem}(See \cite[Theorem 5.4]{GGGlinear} or \cite[Theorem 4.2]{GGGarxiv} )\label{check}
Let $\lambda\in \mathbb{F}_{q^2}$ be a primitive $r$th root of unity. Suppose $\ell={\rm ord}_{rn}(q^2)$.
 Take a primitive $rn$th root of unity $\beta\in \mathbb{F}_{q^{2\ell}}$
such that $\beta^n=\lambda$. Assume that $\mathcal{C}$ is a $\lambda$-constacyclic code of length
$n$ over $\mathbb{F}_{q^2}$ with defining set
$Z=\bigcup_{i=b}^{\delta-2}C_{1+ri}$, where $b$ is a nonnegative integer and $C_{1+ri}$, $b\leq i\leq\delta-2$, are distinct $q^2$-cyclotomic cosets modulo $rn$. Then a parity check matrix of $\mathcal{C}$ can be obtained from the matrix
$$
H_{\mathcal{C}}=\left(
  \begin{array}{ccccc}
    1      &       \beta^{1+rb}  & \beta^{2(1+rb)} & \cdots & \beta^{(n-1)(1+rb)} \\
    1      &       \beta^{1+r(b+1)}  & \beta^{2(1+r(b+1))} & \cdots & \beta^{(n-1)(1+r(b+1))} \\
    \vdots &       \vdots                            & \vdots & \vdots & \vdots \\
    1      &       \beta^{1+r(\delta-3)}   & \beta^{2(1+r(\delta-3))} & \cdots & \beta^{(n-1)(1+r(\delta-3))} \\
     1      &       \beta^{1+r(\delta-2)}   & \beta^{2(1+r(\delta-2))} & \cdots & \beta^{(n-1)(1+r(\delta-2))} \\
  \end{array}
\right)
$$
by expanding each entry as a column vector (containing $\ell$ rows) with respect to certain $\mathbb{F}_{q^2}$-basis of
$\mathbb{F}_{q^{2\ell}}$ and then removing any linearly dependent rows.
\end{lem}

\subsection{Quantum convolutional codes}
A quantum convolutional code is defined through
its stabilizer,  which is a subgroup of the infinite version of
the Pauli group, consisting of tensor products of generalized
Pauli matrices acting on a semi-infinite stream of qudits. The
stabilizer can be defined by a stabilizer matrix of the form
$$
S(D)=\Big(X(D)\,\Big{|}\,Z(D)\Big)\in\mathbb{ F}_q[D]^{(n-k)\times2n}
$$
satisfying $X(D)Z(1/D)^t-Z(D)X(1/D)^t=0$. Let $C$ be a quantum convolutional
code  defined by a full-rank stabilizer matrix
$S(D)$ given above. Then $C$ has parameters $[(n,k,\mu;\gamma,d_f)]_q$,
where $n$ is the frame
size, $k$ is the number of logical qudits per frame,
$\mu=\max_{1\leq i\leq n-k, 1\leq j\leq n}{\{\max{\{\deg X_{ij}(D),\deg Z_{ij}(D)\}}\}}$,
is the
memory,  $d_f$ is the free distance and $\gamma$ is the degree of the
code.

The next result enables us to construct  convolutional
stabilizer codes  from classical convolutional
codes.

\begin{lem}\label{c}
Let $V$ be an $(n,(n-k)/2,\gamma; \mu)_{q^2}$ convolutional code satisfying  $V\subseteq V^{\perp_h}$. Then there exists an $[(n,k,\mu;\gamma,d_f)]_q$
convolutional stabilizer code, where $d_f={\rm wt}(V^{\perp_h}\setminus V)$.
\end{lem}

\begin{lem}(See \cite{AKS} or \cite{GGGlinear})\label{bound}
({\bf Quantum Singleton bound}) The free distance of an $[(n,k,\mu;\gamma,d_f)]_q$, $\mathbb{F}_{q^2}$-linear pure convolutional stabilizer
code is bounded by
$$
d_f\leq\frac{n-k}{2}\Big(\Big\lfloor\frac{2\gamma}{n+k}\Big\rfloor+1\Big)+\gamma+1.
$$
\end{lem}

A quantum convolutional code
achieving  this quantum Singleton bound is called an  {\it maximum-distance-separable (MDS) quantum
convolutional code}.

\section{Code Constructions}
Thereafter, we  always assume that $q$ is an odd prime power.
In this section,  firstly, we  use  constacyclic codes of lengths
$n=q^2+1$ and $n=\frac{q^2+1}{10}$ (assume further that $10\mid(q^2+1)$) respectively to construct classical  convolutional codes.
Consequently,   two classes of MDS quantum convolutional codes are derived from these parameters.

\subsection{MDS quantum convolutional codes  of length $q^2+1$}
The main result of this subsection is Theorem \ref{main1}, which generates a  family of MDS quantum convolutional codes.
The following results are useful to the proof of Theorem \ref{main1}.

\begin{lem}\label{lem1}
Let $n=q^2+1$, $r=q+1$ and $s=1+r\frac{q-1}{2}=\frac{q^2+1}{2}=\frac{n}{2}$, where $q$ is an odd prime power.
Then $\theta_{r,n}=\{1+ri\,|\,0\leq i\leq n-1\}$ is a disjoint union of $q^2$-cyclotomic cosets modulo $rn$:
$$
\theta_{r,n}=C_s\bigcup C_{1+r(\frac{q-1}{2}+\frac{q^2+1}{2})}\bigcup\limits_{i=1}^{s-1}C_{s-ri}
$$
where $C_{s}=\{s\}$, $C_{1+r(\frac{q-1}{2}+\frac{q^2+1}{2})}=\{1+r(\frac{q-1}{2}+\frac{q^2+1}{2})\}$
and $C_{s-ri}=\{s-ri, s+ri\}$ for $1\leq i\leq s-1$.
\end{lem}
\begin{proof}
Note that $rn=(q+1)(q^2+1)$, $rn\nmid(q^2-1)$ and $rn\mid(q^4-1)$, so ${\rm ord}_{rn}(q^2)=2$. We then know that every
$q^2$-cyclotomic coset modulo $rn$ has one or two elements.
A straightforward calculation shows that
$q^2(1+ri)\equiv1+r(q-1-i)\pmod{rn}$ for any integer $i$.
In particular,
$q^2(1+r\frac{q-1}{2})\equiv1+r\frac{q-1}{2}\pmod{rn}$ and
$q^2(1+r(\frac{q-1}{2}+\frac{q^2+1}{2}))\equiv1+r(\frac{q-1}{2}-\frac{q^2+1}{2})\equiv1+r(\frac{q-1}{2}+\frac{q^2+1}{2})\pmod{rn}$, which gives
$$C_{1+r\frac{q-1}{2}}=\Big\{1+r\frac{q-1}{2}\Big\}~~~\hbox{and}~~~C_{1+r(\frac{q^2+1}{2}+\frac{q-1}{2})}=\Big\{1+r(\frac{q-1}{2}+\frac{q^2+1}{2})\Big\}.$$
Clearly, $q^2(s-ri)\equiv s+ri\pmod{rn}$ for any integer $i$.  For $1\leq i\leq s-1$, $s-ri\not\equiv s+ri\pmod{rn}$. Thus
$C_{s-ri}=\{s-ri, s+ri\}$ for $1\leq i\leq s-1$.
It is easy to see that $C_s\neq C_{1+r(\frac{q-1}{2}+\frac{q^2+1}{2})}$.
We want to prove that  the $q^2$-cyclotomic cosets $C_{s-ri}$,  $1\leq i\leq s-1$, are distinct.
Suppose otherwise that two integers $i,j$ with $1\leq i\neq j\leq s-1$ can be found such that
$\{s-ri,s+ri\}=C_{s-ri}=C_{s-rj}=\{s-rj,s+rj\}$. It is obvious that $s-ri\not\equiv s-rj\pmod{rn}$, which forces
$s-ri\equiv s+rj\pmod{rn}$.
This leads to $n\mid(i+j)$, which is a contradiction.
Finally, it is easy to see that the size of the union of these
$q^2$-cyclotomic cosets is equal to $n$. This
completes the proof.
\end{proof}

\begin{lem}\label{contain1}
Let $q$ be an odd prime power and $\lambda\in \mathbb{F}_{q^2}$ be a primitive $(q+1)$th root of unity.
Let $s=\frac{q^2+1}{2}$.
 If $\mathcal{C}$ is a $\lambda$-constacyclic code of length $q^2+1$ over $\mathbb{F}_{q^2}$ with defining set
\begin{equation}\label{constr1}
Z=\bigcup\limits_{j=0}^\delta C_{s-rj}=\Big\{s-r\delta, s-r(\delta-1), \cdots, s-r, s, s+r, \cdots, s+r\delta \Big\},
~~~\hbox{$0\leq\delta\leq\frac{q-1}{2}$},
\end{equation}
then $\mathcal{C}$ is a $[q^2+1,q^2-2\delta, 2\delta+2]$ MDS code  satisfying  $\mathcal{C}^{\perp_{h}}\subseteq \mathcal{C}$.
\end{lem}
\begin{proof}
By Lemma \ref{lem1}, one gets $|Z|=2(\delta+1)-1=2\delta+1$.
We then see that $d(C)=2\delta+2$ by the BCH bound for constacyclic codes (see Lemma \ref{BCH}) and the Singleton bound for linear codes.
It follows  that $\mathcal{C}$ is a $[q^2+1,q^2-2\delta, 2\delta+2]$ MDS code. We  need to show that
$\mathcal{C}^{\perp_{h}}\subseteq \mathcal{C}$.

By Lemma \ref{dual-containong}, it is enough to prove that $Z\bigcap Z^{-q}=\emptyset$. Suppose otherwise that
$Z\bigcap Z^{-q}\neq\emptyset$, i.e. two integers $i,j$ with $0\leq i,j\leq\delta$ can be found such that
$-qC_{s-ri}=C_{s-rj}$. Thus, $-qC_{s-ri}=\{-q(s-ri), -q(s+ri)\}=C_{s-rj}=\{s-rj,s+rj\}$. Two cases may occur at this point:

(i)~$-q(s-ri)\equiv s-rj\pmod{rn}$. After expanding and reducing this equation, we obtain $qi+j\equiv s\pmod{n}$.
Since $0\leq qi+j\leq q\delta+\delta\leq q\frac{q-1}{2}+\frac{q-1}{2}=\frac{q^2-1}{2}<n$ and $0<s<n$,
it follows that $qi+j=s$. However, $qi+j\leq\frac{q^2-1}{2}<s=\frac{q^2+1}{2}$. This is a contradiction.

(ii)~$-q(s-ri)\equiv s+rj\pmod{rn}$. Similarly, we obtain $qi\equiv s+j\pmod{n}$.
Clearly, $0\leq qi\leq\frac{q^2-q}{2}<n$ and $0<s+j<n$. Thus $qi=s+j$. However, $qi\leq\frac{q^2-q}{2}<s+j$. This is a contradiction.
\end{proof}

Using Lemma \ref{check}, Theorem \ref{juanji} and Lemma \ref{lem1},
we obtain the following  classical convolutional codes.
\begin{lem}\label{classical}
Let $n=q^2+1$, where $q\geq5$ is an odd prime power. Let $i$ be an integer with  $2\leq i\leq \frac{q-1}{2}$.   Then there exists  a classical convolutional code $V$ with parameters
$(n,2i-1, 2; 1, \geq n-2i)_{q^2}$;
the free distance of $V^{\perp_h}$ is exactly equal to $2i+2$. Furthermore, $V$ satisfies $V\subseteq V^{\perp_h}$.
\end{lem}
\begin{proof}
Let $r=q+1$ and $\lambda\in\mathbb{ F}_{q^2}$ be a primitive $r$th root of unity.
Assume that $\beta$ is a primitive $rn$th root of unity  in some extension field of $\mathbb{F}_{q^2}$ such that $\beta^n=\lambda$.
Since ${\rm ord}_{rn}(q^2)=2$, it follows that $\beta\in\mathbb{ F}_{q^4}$.
Let $s=\frac{q^2+1}{2}$ and $i$ be an integer with $2\leq i\leq (q-1)/2$. Let $\mathcal{C}$ be a $\lambda$-constacyclic code of length $n$ over $\mathbb{F}_{q^2}$ with defining set $Z=\bigcup_{b=0}^{i} C_{s-rb}$.
It follows from Lemma \ref{check} that a parity check matrix  of $\mathcal{C}$, denoted by $N_\mathcal{C}$,  can be obtained from the following matrix
$$
H_{\mathcal{C}}=\left(
  \begin{array}{ccccc}
    1      &       \beta^{s}  & \beta^{2s} & \cdots & \beta^{(n-1)s} \\
    1      &       \beta^{s-r}  & \beta^{2(s-r)} & \cdots & \beta^{(n-1)(s-r)} \\
    \vdots &       \vdots                            & \vdots & \vdots & \vdots \\
    1      &       \beta^{s-r(i-1)}   & \beta^{2(s-r(i-1)))} & \cdots & \beta^{(n-1)(s-r(i-1))} \\
     1      &       \beta^{s-ri}   & \beta^{2(s-ri)} & \cdots & \beta^{(n-1)(s-ri)} \\
  \end{array}
\right)
$$
by expanding each entry as a column vector (containing $2$ rows) with respect to certain $\mathbb{F}_{q^2}$-basis of
$\mathbb{F}_{q^4}$ and then removing any linearly dependent rows.
Therefore, $N_\mathcal{C}$ has rank $2i+1$, implying that $\mathcal{C}$ is an MDS code with parameters  $[n, n-2i-1, 2i+2]$.
Consequently, $\mathcal{C}^{\perp_h}$ is also an MDS code with parameters  $[n, 2i+1, n-2i]$.

Now let $\mathcal{C}_0$ be a $\lambda$-constacyclic code of length $n$ over $\mathbb{F}_{q^2}$ with defining set $Z_0=\bigcup_{b=0}^{i-1} C_{s-rb}$. Similar reasoning shows that $\mathcal{C}_0$ is an MDS code with  parameters  $[n, n-2i+1, 2i]$, and that $\mathcal{C}^{\perp_h}_0$ is an
MDS code with  parameters  $[n, 2i-1, n-2i+2]$.
Further, a parity check matrix  of $\mathcal{C}_0$, denoted by $N_{\mathcal{C}_0}$,  can be obtained from the following matrix
$$
H_{\mathcal{C}_0}=\left(
  \begin{array}{ccccc}
    1      &       \beta^{s}  & \beta^{2s} & \cdots & \beta^{(n-1)s} \\
    1      &       \beta^{s-r}  & \beta^{2(s-r)} & \cdots & \beta^{(n-1)(s-r)} \\
    \vdots &       \vdots                            & \vdots & \vdots & \vdots \\
    1      &       \beta^{s-r(i-1)}   & \beta^{2(s-r(i-1)))} & \cdots & \beta^{(n-1)(s-r(i-1))} \\
  \end{array}
\right)
$$
by expanding each entry as a column vector (containing $2$ rows) with respect to the  $\mathbb{F}_{q^2}$-basis of
$\mathbb{F}_{q^4}$ and then removing any linearly dependent rows (This has been done, since $H_{\mathcal{C}_0}$ is a submatrix of $H_{\mathcal{C}}$).
In particular, $N_{\mathcal{C}_0}$ has rank $2i-1$.

Next  let $\mathcal{C}_1$ be a $\lambda$-constacyclic code of length $n$ over $\mathbb{F}_{q^2}$ with defining set $Z_0=C_{s-ri}$.
Thus $\mathcal{C}_1$ has parameters $[n,n-2, \geq2]$. A parity check matrix, denoted by $N_{\mathcal{C}_1}$, is given by expanding the entries of the matrix
$$
H_{\mathcal{C}_1}=\Big[1, \beta^{s-ri}, \beta^{2(s-ri)}, \cdots, \beta^{(n-1)(s-ri)}\Big]
$$
with respect to $\beta$ (This has been done, since $H_{\mathcal{C}_1}$ is a submatrix of $H_{\mathcal{C}}$).
According to Theorem \ref{juanji} (1), a convolutional code $V$ is obtained which is generated by the reduced basic
generator matrix
$$
G(D)=\tilde{N}_{\mathcal{C}_0}+\tilde{N}_{\mathcal{C}_1}D
$$
where $\tilde{N}_{\mathcal{C}_0}=N_{\mathcal{C}_0}$ and $\tilde{N}_{\mathcal{C}_1}$ is derived from $N_{\mathcal{C}_1}$
by adding zero-rows at the bottom such that the rows of $\tilde{N}_{\mathcal{C}_1}$ is exactly equal to  the number of rows of $N_{\mathcal{C}_0}$.
It follows from Theorem \ref{juanji} that $V$ is a  convolutional code of dimension $2i-1$,  degree $2$, memory $1$ and free distance $\geq n-2i$.
For the free distance of $V^{\perp_h}$, we have that $\min{\{\geq2i+2,2i+2\}}\leq d_f^{\perp_h}\leq2i+2$ which forces $d_f^{\perp_h}=2i+2$.

Finally, it follows from Lemma \ref{contain1} that $\mathcal{C}^{\perp_h}\subseteq \mathcal{C}$, which gives
$V^{\perp_h}\subseteq V$ by Theorem \ref{juanji} (2). This completes the proof.
\end{proof}

We are now in a position to  show the main result of this
subsection.
\begin{Theorem}\label{main1}
Let $n=q^2+1$, where $q\geq5$ is an odd prime power.
Let $i$ be an integer with  $2\leq i\leq(q-1)/2$.
Then there exist MDS quantum convolutional codes with parameters
$[(n,n-4i+2,1;2,2i+2)]_q$.
\end{Theorem}
\begin{proof}
By Lemma \ref{classical}, we have constructed a  convolutional code $V$ with parameters $(n,2i-1, 2;1, \geq n-2i)_{q^2}$; furthermore, $V$ satisfies
$V\subseteq V^{\perp_h}$. Now $n=q^2+1$, $\gamma=2$ and $\mu=1$. Let $k$ be an integer satisfying $\frac{n-k}{2}=2i-1$. Thus $k=n-4i+2$.
Note that ${\rm wt}(V^{\perp_h})=2i+2$ and ${\rm wt}(V)\geq n-2i$.
  It is clear that $n-2i>2i+2$, which shows $d_f={\rm wt}(V^{\perp_h}\setminus V)=2i+2$.  Using Lemma \ref{c}, there exists
an $[(n,n-4i+2,1;2,2i+2)]_q$ convolutional stabilizer code.
Finally, we show that the resulting convolutional stabilizer code attains the Quantum Singleton bound (see Lemma \ref{bound}):
$$
\frac{n-k}{2}\Big(\Big\lfloor\frac{2\gamma}{n+k}\Big\rfloor+1\Big)+\gamma+1=(2i-1)\cdot(0+1)+2+1=2i+2=d_f.
$$
\end{proof}

\begin{Example}
In Table $1$, we list some MDS quantum convolutional codes   obtained from
Theorem~\ref{main1} for $q=7, 11, 13, 19$ and $23$.

\begin{center}
\begin{longtable}{c|c|c}  
\caption{MDS Quantum  Convolutional Codes}\\\hline
$q$    &    $[(q^2+1, q^2-4i+3,1;2, 2i+2)]_q$  &   $2\leq i\leq\frac{q-1}{2}$   \\\hline
$7$    &   $[(50, 52-4i,1;2, 2i+2)]_7$   & $2\leq i\leq 3$ \\
$11$    &   $[(122, 124-4i,1;2, 2i+2)]_{11}$   & $2\leq i\leq 5$ \\
$13$    &   $[(170, 172-4i,1;2, 2i+2)]_{13}$   & $2\leq i\leq 6$ \\
$19$      &   $[(362, 364-4i,1;2, 2i+2)]_{19}$   & $2\leq i\leq 9$ \\
$23$     &   $[(530, 532-4i,1;2, 2i+2)]_{23}$   & $2\leq i\leq 11$ \\
\hline
\end{longtable}
 \end{center}

\end{Example}

\subsection{MDS quantum convolutional codes  of length $\frac{q^2+1}{10}$}

Let $q$ be an odd prime power such that $10\mid(q^2+1)$, i.e., $q$ has  the form $10m+3$ or $10m+7$, where $m$ is a positive integer.
Let $n=\frac{q^2+1}{10}$, $s=\frac{q^2+1}{2}$ and $r=q+1$. It is clear that $s\equiv1\pmod{r}$, which implies that
$s\pmod{rn}\in\theta_{r,n}=\{1+ri\,|\,0\leq i\leq n-1\}$.
As  in the previous subsection, we need the following lemmas.

\begin{lem}\label{cyclotomic}
Assume that $q$ is an odd prime power with $10\mid(q^2+1)$.
Let $n=\frac{q^2+1}{10}$, $s=\frac{q^2+1}{2}$ and $r=q+1$. Then $\theta_{r,n}=\{1+ri\,|\,0\leq i\leq n-1\}$ is
a disjoint union of $q^2$-cyclotomic cosets modulo $rn$:
$$
\theta_{r,n}=C_{s}\bigcup\Big(\bigcup\limits_{k=0}^{\frac{n-1}{2}-1}C_{s-(q+1)(\frac{n-1}{2}-k)}\Big).
$$
\end{lem}
\begin{proof}
Observe that $q^4\equiv1~(\bmod~rn)$, which implies that each $q^2$-cyclotomic coset modulo $rn$ contains one or two elements.
Now,
$$
q^2\big(1+(q+1)j\big)=q^2+q^2(q+1)j=q^2+(q^2+1-1)(q+1)j\equiv q^2-(q+1)j~(\bmod~rn).
$$
It is clear that for $0\leq j\leq n-1$, $1+(q+1)j\equiv q^2-(q+1)j~(\bmod~rn)$ if and only if $j=\frac{q-1}{2}+dn$ ($d$ is an integer),
which forces $d=0$ and hence $j=\frac{q-1}{2}$. This shows that $s=1+(q+1)\frac{q-1}{2}=\frac{q^2+1}{2}$ is the unique element of $\theta_{r,n}$ with
$q^2s\equiv s~(\bmod~rn)$.
To complete the proof, it suffices to show that for any $0\leq i\neq j\leq\frac{n-1}{2}-1$,
$C_{s-(q+1)(\frac{n-1}{2}-i)}=\{s-(q+1)(\frac{n-1}{2}-i),s+(q+1)(\frac{n-1}{2}-i)\}$
and $C_{s-(q+1)(\frac{n-1}{2}-j)}$ are distinct.
Suppose otherwise that $C_{s-(q+1)(\frac{n-1}{2}-i)}=C_{s-(q+1)(\frac{n-1}{2}-j)}$ for some $0\leq i\neq j\leq\frac{n-1}{2}-1$.
If $s-(q+1)(\frac{n-1}{2}-i)\equiv s-(q+1)(\frac{n-1}{2}-j)~(\bmod~rn)$,  then $i\equiv j~(\bmod~n)$ which is impossible;
If  $s-(q+1)(\frac{n-1}{2}-i)\equiv s+(q+1)(\frac{n-1}{2}-j)~(\bmod~rn)$, then $i+j\equiv -1~(\bmod~n)$ which is a contradiction.
\end{proof}

Let $\lambda\in\mathbb{ F}_{q^2}$ be a primitive $r$th root of unity, and let $\beta\in \mathbb{F}_{q^4}$ be a primitive $rn$th
root of unity such that $\beta^n=\lambda$.
Let $\mathcal{C}$ be a $\lambda$-constacyclic code of length $n=\frac{q^2+1}{10}$ over $\mathbb{F}_{q^2}$
with defining set
\begin{equation}\label{set}
Z=\bigcup\limits_{j=0}^{2m-1}C_{s-(q+1)(\frac{n-1}{2}-j)}.
\end{equation}
We then know from Lemma~\ref{cyclotomic} that $Z$ is a disjoint union of $q^2$-cyclotomic cosets modulo $rn$ with $|Z|=4m$.
Moreover, we assert that the minimum distance of $\mathcal{C}$ is exactly equal to $4m+1$. To see this, observe that
$$
Z=\big\{s+r(\frac{n-1}{2}-2m+1), s+r(\frac{n-1}{2}-2m+2),\cdots, s+r(\frac{n-1}{2}-1),$$ $$
s+r\frac{n-1}{2}, s-r\frac{n-1}{2}, s-r(\frac{n-1}{2}-1),\cdots, s-r(\frac{n-1}{2}-2m+1)\big\}.
$$
A simple calculation shows that $s+r\frac{n-1}{2}+r\equiv s-r\frac{n-1}{2}\pmod{rn}$.
By the BCH bound for constacyclic  codes, $\mathcal{C}$ is an MDS code with parameters $[n, n-4m, 4m+1]$.

The next result shows  that $\mathcal{C}$ is a dual-containing code.
\begin{lem}
Assume that $q$ is an odd prime power with   the form $10m+3$ or $10m+7$, where $m$ is a positive integer.
Let $n=\frac{q^2+1}{10}$, $s=\frac{q^2+1}{2}$ and $r=q+1$. Let $\mathcal{C}$ be a $\lambda$-constacyclic code of length $n=\frac{q^2+1}{10}$ over $\mathbb{F}_{q^2}$
with defining set as in (\ref{set}), where $\lambda\in \mathbb{F}_{q^2}$ is a primitive $r$th root of unity.
Then $\mathcal{C}$ is a dual-containing code.
\end{lem}
\begin{proof}
We have to prove that $Z\bigcap Z^{-q}=\emptyset$. We just give a proof for the case $q=10m+3$. The  case for $q=10m+7$ is proved  similarly.
Suppose there exist integers $j,k$ with $0\leq j,k\leq 2m-1$ such that $C_{-q(s-(q+1)(\frac{n-1}{2}-j))}=C_{s-(q+1)(\frac{n-1}{2}-k)}$.
Write $j=j_1m+j_0$ and $k=k_1m+k_0$, where $j_1,k_1\in\{0,1\}$ and $0\leq j_0,k_0<m$. Let  $j_0'=m-j_0$ and $k_0'=m-k_0$, and so $0<j_0', k_0'\leq m$.

{\it Case I.}~~$-q(s-(q+1)(\frac{n-1}{2}-j))\equiv s-(q+1)(\frac{n-1}{2}-k)~(\bmod~rn)$.
After routine computations, we obtain
\begin{equation}\label{one}
-\frac{q+1}{2}\equiv qj+k~(\bmod~n).
\end{equation}
Now $qj+k=(10m+3)(j_1m+j_0)+k_1m+k_0=10j_1m^2+(10j_0+3j_1+k_1)m+3j_0+k_0=10j_1m^2+(10m-10j_0'+3j_1+k_1)m+3m-3j_0'+m-k_0'$.

 Assume
$qj+k<n$.

If $j_1=0$, it follows from (\ref{one}) that
$$
(10m-10j_0'+k_1)m+3m-3j_0'+m-k_0'=n-\frac{q+1}{2}=10m^2+m-1.
$$
This leads to
$$
(k_1-10j_0'+4)m=m+3j_0'+k_0'-1,
$$
which is a contradiction,  since $(k_1-10j_0'+4)m<0$ and $m+3j_0'+k_0'-1>0$.

If $j_1=1$, then
$$
10m^2+(10m-10j_0'+3+k_1)m+3m-3j_0'+m-k_0'=10m^2+m-1,
$$
or equivalently,
$
10j_0'm+k_0'=1+(k_1+10m+6)m-3j_0'.
$
Now, $10j_0'm+k_0'\leq10m^2+m$,  but $1+(k_1+10m+6)m-3j_0'>10m^2+3m$, which is a contradiction.

Assume
$qj+k>n$.

If $j_1=0$, then $qj+k=(10m+3)j_0+k<(10m+3)m+2m=10m^2+5m<n$; this is impossible.

If $j_1=1$, we claim that $qj+k-n=(k_1+10j_0-3)m+k_0+3j_0-1<n$; this is because
$(k_1+10j_0-3)m+k_0+3j_0-1\leq(1+10m-10-3)m+m-1+3(m-1)-1<n$. From (\ref{one}) again,
we have $(k_1+10j_0-3)m+k_0+3j_0-1=n-\frac{q+1}{2}=10m^2+m-1$, or equivalently,
$(k_1-10j_0')m=k_0'+3j_0'$. This is a contradiction, because $k_0'>0, j_0'>0$
and $k_1-10j_0'<0$.

{\it Case II.}~~$-q(s-(q+1)(\frac{n-1}{2}-j))\equiv s+(q+1)(\frac{n-1}{2}-k)~(\bmod~rn)$.
After routine computations, we get
\begin{equation}\label{two}
-\frac{q-1}{2}\equiv qj-k~(\bmod~n).
\end{equation}
As we did previously, $qj-k=(10m+3)(j_1m+j_0)-k_1m-k_0=10j_1m^2+(10j_0+3j_1-k_1)m+3j_0-k_0$.

If $j_1=0$, then $qj-k\leq(10m+3)(m-1)<n$. When $0<qj-k<n$,  by (\ref{two}), $10j_0m+3j_0-k_1m-k_0=10m^2+m$,
which is equivalent to $10j_0'm-m+3j_0'+k_1m-k_0'=0$. This is impossible, since
$10j_0'm-m+3j_0'+k_1m-k_0'>10m-m-m>0.$
When $qj-k<0$ (Clearly, $0<k-qj<n$), we obtain $5m+1=\frac{q-1}{2}=k-qj$, which is a contradiction since $k<2m$.

If $j_1=1$ and $j_0=0$, we have  $qj-k=(10m+3)m-k<n$.
Using (\ref{two}), we get $k=2m$, also a contradiction.

If $j_1=1$ and $j_0>0$, we then know that $qj-k=10m^2+(10j_0+3-k_1)m+3j_0-k_0>n$.
On the other hand, $qj-k-n=10m^2-10j_0'm-3j_0'-k_1m-m+k_0'-1<n.$
Applying (\ref{two}) again, we obtain $-10j_0'm-3j_0'-k_1m-2m+k_0'-1=0$.
This is impossible, because $-10j_0'm-3j_0'-k_1m-2m+k_0'-1<0$.
\end{proof}

 The proof of next lemma is quite similar to that of Lemma \ref{classical}, so we omit its proof.
\begin{lem}\label{classical2}
Assume that $q$ is an odd prime power with   the form $10m+3$ or $10m+7$, where $m\geq2$ is a positive integer.
Let $i$ be an integer with  $2\leq i\leq 2m-1$ (This requires $m\geq2$).   Then there exists  a classical convolutional code $V$ with parameters
$(n,2i, 2; 1, \geq n-2i-1)_{q^2}$;
the free distance of $V^{\perp_h}$ is exactly equal to $2i+3$. Furthermore, $V$ satisfies $V\subseteq V^{\perp_h}$.
\end{lem}

Combining Lemma \ref{c} with Lemma \ref{classical2}, we obtain the following result.

\begin{Theorem}\label{main2}
Assume that $q$ is an odd prime power with   the form $10m+3$ or $10m+7$, where $m\geq2$ is a positive integer.
Let $n=\frac{q^2+1}{10}$ and $i$ be an integer with  $2\leq i\leq 2m-1$ (This requires $m\geq2$).
Then there exist MDS quantum convolutional codes with parameters
$[(n,n-4i,1;2,2i+3)]_q$.
\end{Theorem}
\begin{proof}
By Lemma \ref{classical2}, we have constructed a  convolutional code $V$ with parameters $(n,2i, 2;1, \geq n-2i+1)_{q^2}$; furthermore, $V$ satisfies
$V\subseteq V^{\perp_h}$. Now $n=\frac{q^2+1}{10}$, $\gamma=2$ and $\mu=1$. Let $k$ be an integer satisfying $\frac{n-k}{2}=2i$. Thus $k=n-4i$.
Note that ${\rm wt}(V^{\perp_h})=2i+3$ and ${\rm wt}(V)\geq n-2i-1$.
Since $n-2i-1>2i+3$, which gives  $d_f={\rm wt}(V^{\perp_h}\setminus V)=2i+3$.  Using Lemma \ref{c}, there exists
an $[(n,n-4i,1;2,2i+3)]_q$ convolutional stabilizer code.
Finally, we show that the resulting convolutional stabilizer code attains the Quantum Singleton bound (see Lemma \ref{bound}):
$$
\frac{n-k}{2}\Big(\Big\lfloor\frac{2\gamma}{n+k}\Big\rfloor+1\Big)+\gamma+1=2i\cdot(0+1)+2+1=2i+3=d_f.
$$
\end{proof}

\begin{Example}
In Table $2$, we list some MDS quantum convolutional codes   obtained from
Theorem~\ref{main2}.

\begin{center}
\begin{longtable}{c|c|c|c}  
\caption{MDS Quantum  Convolutional Codes}\\ \hline
$m$    &$q$&    $[((q^2+1)/10, (q^2+1)/10-4i,1;2, 2i+3)]_q$  &   $2\leq i\leq2m-1$   \\\hline
$2$ &   $23$    &   $[(53, 53-4i,1;2, 2i+3)]_{23}$   & $2\leq i\leq 3$ \\
$2$ &$27$    &   $[(73, 73-4i,1;2, 2i+3)]_{27}$   & $2\leq i\leq 3$ \\
$3$ &$13$    &   $[(137, 137-4i,1;2, 2i+3)]_{13}$   & $2\leq i\leq 5$\\
\hline
\end{longtable}
 \end{center}
\end{Example}

\vspace*{0.6 cm}
\noindent{\bf Acknowledgements}
The first author is supported by NSFC   (Grant No.~11171370),
the Youth Backbone Teacher Foundation of  Henan's University (Grant No.~2013GGJS-152) and,
Science and Technology Development Program of Henan Province in 2014 (144300510051).
The research of the  second author is partially supported by NSFC   (Grant No.~11271005) and Nanyang Technological University's research grant number M4080456.

\end{document}